\documentclass[12pt,a4paper]{article}
\usepackage[a4paper, margin=1in]{geometry}
\usepackage[normalem]{ulem}
\usepackage{lmodern}
\usepackage[english]{babel}
\usepackage[utf8]{inputenc}
\usepackage[T1]{fontenc}
\usepackage{pdfpages}
\usepackage{microtype}
\usepackage{amsthm}
\usepackage{xcolor}
\usepackage{hyperref}
\usepackage{bm}
\usepackage{amsmath, amsfonts, amssymb}
\usepackage{authblk}
\usepackage{setspace}
\usepackage{framed}
\usepackage{cancel}
\begin{document}
\newcommand{\eqrefc}[1]{\textcolor{black}{\eqref{#1}}}
\newtheorem{theorem}{Proposition}
\begin{center}
	\textbf{{\Large Joint Momenta-Coordinates States as  Pointer States \\
			\vspace{0.5cm}
			in Quantum Decoherence. }}
	\vspace{1cm}\\
\end{center}
\textbf{Nomenjanahary Tanjonirina Manampisoa$ ^{1} $, Ravo Tokiniaina Ranaivoson$^{2}$,\\ Roland Raboanary$^{3}$, Raoelina Andriambololona$^{4}$, Rivo Herivola Manjakamanana Ravelonjato$^{5}$, Naivo Rabesiranana$ ^{6} $.} \vspace{0.5cm}\\
\textit{ tanjonaphysics@gmail.com}$ ^{1} $\\
 \textit{tokhiniaina@gmail.com$ ^{2} $, tokiniainarvor13@gmail.com}$ ^{2} $\\
 \textit{r\_raboanary@yahoo.fr}$ ^{3} $\\
 \textit{raoelina.andriambololona@gmail.com}$ ^{4} $\\
 \textit{manjakamanana@yahoo.fr}$ ^{5} $\\
 \textit{rabesiranana@yahoo.fr}$ ^{6} $
 \vspace{0.5cm}\\
$ ^{1,2,4,5,6}$ Institut National des Sciences et Techniques Nucléaires (INSTN-Madagascar)\\
 		BP 3907 Antananarivo 101, Madagascar,
 	\textit{instn@moov.mg}\vspace{0.3cm}\\
 $ ^{2,4} ${\textit{TWAS Madagascar Chapter, Malagasy Academy,\\BP 4279 Antananarivo 101, Madagascar}}\vspace{0.3cm}\\
 $ ^{1,3} ${\textit{ Faculty of Sciences, iHEPMAD-University of Antananarivo,\\ BP 566 Antananarivo 101, Madagascar}}

	\begin{abstract}
		 Quantum decoherence provides a framework to study the emergence of classicality from quantum systems by showing how interactions with the environment suppress interferences and select robust states known as pointer states. Earlier studies have linked Gaussian coherent states to  pointer states. More recently, it was conjectured that more general quantum states called joint momenta-coordinates states may serve  as more suitable candidates to be pointer states. These states are associated to the concept of quantum phase space and saturate, by definition,   generalized uncertainty relations. In this work, we rigorously prove this conjecture. Building on the Lindblad framework for the damped harmonic oscillator, and applying Zurek's predictability-sieve criterion, we analyze both underdamped and overdamped regimes. We show that only in the underdamped case do joint momenta-coordinates states remain pure and robust for all times, establishing them as the true pointer states. This extends Isar's earlier underdamped treatment, generalizes the concept beyond Gaussian approximations, and embeds classical robustness in the quantum phase space formalism, with potential applications in error-resistant quantum information.\\

		
	\end{abstract}

 \newpage		 
\section{Introduction}  
\hspace{0.80cm}
Quantum mechanics is one of the most successful theories in physics, with predictive power that has led to revolutionary technologies such as quantum computing. Despite this success, the emergence of classical behavior from quantum laws remains a central puzzle in the foundations of physics \cite{ Zurek, foundations}. For nearly a century, researchers have debated how definite outcomes arise from quantum superpositions. 
One of the leading approaches to this problem is quantum decoherence, which explains how the interaction of a system with its environment suppresses interference between different quantum states. Decoherence thus aims to provide  a mechanism for the transition from quantum to classical while staying within the conventional framework of quantum mechanics \cite{zeh, joos}. A key concept in this approach is that of pointer states: those special states that remain robust under environmental influence and correspond most closely to classical outcomes \cite{zurek1,zurek2}.

 \par
In classical mechanics, a particle is fully described by its position $ x $ and its momentum $ p $. In quantum mechanics, however, the Heisenberg uncertainty principle prevents these quantities from being specified simultaneously with arbitrary precision.  Various phase space formulations in quantum mechanics have been developed to address this limitation, including the Wigner function. However,  the phase space considered in this formulation is a classical one \cite{Weigner, Bondar}. More recently, the idea of quantum phase space based on joint momenta-coordinates states has been proposed \cite{ravo2,invariant,linear}. These states saturate generalized uncertainty relations and share some properties with classical states.\par
Previous studies, most notably those by Zurek and collaborators \cite{Zurek, zurek1,zurek2}, have shown that Gaussian coherent states often emerge as approximate pointer states because they minimize uncertainty relation and decohere more slowly than other superpositions. However, this identification is model-dependent and remains tied to a particular coordinate or momentum basis in Hilbert space.
 Others studies  applied the Lindblad formalism to the damped harmonic oscillator and provided important insights into the role of decoherence in open quantum systems \cite{scutaru, isar, ref2, ref3, ref4, ref5, ref6, revisiting}. A comprehensive treatment of the damped harmonic oscillator, covering both underdamped and overdamped regimes, was provided in the seminal work \cite{scutaru}. However, the crucial question of pointer states was not addressed in that work. While pointer states was discussed in the work of Isar in \cite{isar}, their analysis focused mainly on the underdamped regime, where the dynamics allows for Gaussian-like states that remain approximately stable under environmental influence. However, the complementary overdamped case, where dissipation overwhelms the restoring force, was not treated in their work.  A recent study \cite{highlighting} put forward the conjecture that  pointer states may, in fact, coincide with joint momenta-coordinates states, but left this connection without a formal proof.

In this present work, we provide the first rigorous proof to this conjecture. Specifically:
\begin{itemize}
	\item We extend  Isar's analysis \cite{isar} to both underdamped and overdamped cases, showing that only in the underdamped regime do pure states survive as robust pointer states. 
	\item We explain physically why in the overdamped case no pointer states exist.
	\item We show that the predictability-sieve criterion of Zurek, which identifies states least affected by decoherence, is satisfied precisely by joint momenta-coordinates states.

\end{itemize}
This result not only confirms the conjecture raised in \cite{highlighting}, but also strengthens the conceptual bridge between decoherence theory and quantum phase space.\\ In other  models such as for a free particle undergoing quantum Brownian motion, pointer states cannot be well-defined; in this case, the predictability sieve does not select well defined pointer states \cite{sieve}. For this reason, in this work we choose the damped harmonic oscillator model where the predictability sieve criterion can define the pointer states unambiguously. 
 \par  
This paper is structured as follows: in Section 2, we briefly review the concept of quantum decoherence and pointer states. We portray the concept of joint momenta-coordinates and quantum phase space in Section 3. In Section 4, we describe the Lindblad model for open quantum system and analyze the time evolution of the covariance matrix. In Section \eqrefc{section 5}, we prove that joint momenta-coordinates satisfy the predictability sieve criterion  and identify them as pointer states. Discussions and  conclusions are given in Section 6. Bold letters like $ \mathbf{p} $ are used to represent quantum operators, while normal letters correspond to eigenvalues or non-operator quantities. 

\section{Decoherence and pointer states}
Quantum decoherence theory is a quantitative model of how the transition from quantum to classical mechanics occurs, which involves systems interacting with an environment that effectively performs measurements \cite{decoherence}. Decoherence, which is the loss of definite phase relations between the different components of a superposition state, arises from the interaction of the system with an environment.\par
Decoherence is usually described using a Hilbert space $ \mathcal{ H} $ composed of three parts: the subsystem $ \mathcal{S} $ (the system), the measuring apparatus $ \mathcal{A} $  and the environment $\mathcal{E}$, such that
\begin{equation}
\mathcal{H}= \mathcal{H_{S}}\otimes \mathcal{H_{A}}\otimes \mathcal{H_{E}}  \, \, .
\end{equation} 
In an idealized scenario, the total state initially factorizes as a superposition of system states $ |s_{n}\rangle $ correlated with an initial apparatus state $ |a_{0}\rangle $ and environment state $ |e_{0}\rangle $ \cite{maximilian}
\begin{equation}
|\psi(0)\rangle= \left( \sum_{n}c^{n}|s_{n}\rangle\right)\otimes |a_{0}\rangle \otimes |e_{0}\rangle  = \left( \sum_{n}c^{n}|s_{n}\rangle\right) |a_{0}\rangle |e_{0}\rangle  \, \, .
\end{equation}
As the system interacts with the apparatus, correlations develop
\begin{equation}
|\psi(t_{1})\rangle=\left( \sum_{n}c^{n}|s_{n}\rangle |a_{n}(t_{1})\rangle\right)|e_{0}\rangle  \, \, .
\end{equation}
Further interaction with the environment leads to entanglement between all three components:
\begin{equation}
|\psi(t)\rangle= \sum_{n}c^{n}|s_{n}\rangle |a_{n}(t)\rangle |e_{n}(t)\rangle, \quad 0 < t_{1} < t \, \, .
\end{equation}
According to \cite{pointers}, we must make a few assumptions to establish such  dynamics. First, we assume that the interaction Hamiltonian governs the time evolution of the state $ |\psi \rangle  $. Next, we consider the part of the interaction Hamiltonian that couples the system $ \mathcal{S} $ to the apparatus $ \mathcal{A} $, and we assume that the system basis states $ |s_{n}\rangle $ remain unchanged except for a multiplicative factor (the same assumption applies to the interaction between the apparatus and the environment and the apparatus states $ |a_{n}(t) \rangle$). We also assume two distinct time scales: one for the system-apparatus, and another for the apparatus-environment interaction, with these processes considered to occur in sequence. In a more realistic setting, these assumptions may not hold strictly, but they still offer a useful approximation. As discussed further below, the general goal is to identify the apparatus states that best preserve their correlations with the system over time.\par
Since the environment is not directly observed, we trace over its degrees of freedom. This yields the reduced density matrix for the system-apparatus $ \mathcal{S}\mathcal{A} $ subsystem:
\begin{align}
\bm{\rho}_{SA}(t)= &Tr_{\mathcal{E}}|\Psi(t)\rangle \langle \Psi (t)| \notag \\ 
=&\sum_{n,m}c^{n}c_{m}^{*}\left\langle e^{m}(t) |e_{n}(t)\right\rangle |s_{n}\rangle  |a_{n}(t)\rangle \langle  s^{m}|\langle a^{m}(t)|  \, \, .  
\end{align}
A key mathematical claim of decoherence is that the states $ \left\lbrace |e_{n}(t)\rangle\right\rbrace  $ will rapidly become orthogonal to each other so that
\begin{equation}
\langle e^{m}(t)|e_{n}(t)\rangle \rightarrow \delta^{m}_{n}=\left\lbrace 
\begin{aligned}
1\quad if \quad n &=m\\
0 \quad if \quad n &\neq m \\
\end{aligned}
\right. \\ \, \, .
\end{equation}
In this limit, the reduced density matrix becomes diagonal in the basis  $  |s_{n}\rangle |a_{n}(t)\rangle $\cite{understanding}:
\begin{equation}
\bm{\rho}_{SA}(t)\rightarrow \sum_{n}|c^{n}|^{2}|s_{n}\rangle |a_{n}(t)\rangle\langle s^{n}|\langle a^{n}(t)| \, \, .
\end{equation}
The  collection of states $ \left\lbrace |s_{n}\rangle |a_{n}(t)\rangle \right\rbrace  \subset \mathcal{H_{S}} \otimes \mathcal{H_{A}}$ that retain their correlations with the system despite environmental disturbance are known as      pointer states. These are the states which get minimally affected by decoherence whereas a superposition of such states generally loses its quantum coherence quite rapidly. They are robust, stable against decoherence and correspond to classical outcomes\cite{venugopalan}. \\
This last equation indicates that $ \bm{\rho}_{SA}(t) $ becomes diagonal in (a subset of) pointer states. The interaction with the environment has suppressed interference terms of the type $ |s_{n}\rangle |a_{n}(t)\rangle \langle s^{m}(t)|a^{m}(t)|, m\neq n $, that would have been present in the $ \mathcal{ SA} $ density matrix if $ \mathcal{ A} $ were not coupled to $ \mathcal{ E} $.

 Most models of  decoherence adopt an environment represented by a large collection of harmonic oscillators (bath of harmonic oscillators) with which the system interacts, often through a coordinate-coordinate coupling \cite{anu}. This simple yet powerful model captures the essential mechanism by which environmental noise drives the selection of pointer states.

\section{Joint momenta-coordinates states and quantum phase space }
To construct a quantum phase space compatible with the uncertainty principle, we must first define its fundamental algebra. We may consider     
 the general multidimensional and relativistic  Canonical Commutation Relations (CCRs), characterized by a metric signature ($ D_{+},D_{-} $)  with a number of pair ($ \left[\mathbf{ p}_{\mu}, \mathbf{x}_{\nu}\right]  $) equal to $ D= D_{+} + D_{-}    $ i.e $ \mu=0,1,2,\cdots,D-1 $, which are given by \cite{highlighting, linear}:
\begin{equation}
\left\lbrace 
\begin{aligned}
\left[\mathbf{ p}_{\mu}, \mathbf{x}_{\nu}\right]  &=\mathbf{ p}_{\mu} \mathbf{x}_{\nu} - \mathbf{x}_{\nu}\mathbf{p}_{\mu}=i\hbar\eta_{\mu \nu}\\
\left[\mathbf{ p}_{\mu}, \mathbf{p}_{\nu}\right]  &=\mathbf{ p}_{\mu} \mathbf{p}_{\nu} - \mathbf{p}_{\nu}\mathbf{p}_{\mu}=0\\
\left[\mathbf{ x}_{\mu}, \mathbf{x}_{\nu}\right] &=\mathbf{ x}_{\mu} \mathbf{x}_{\nu} - \mathbf{x}_{\nu}\mathbf{x}_{\mu}=0  \, \, ,
\end{aligned}
\right.
\end{equation}
with 
\begin{equation}
\eta_{\mu \nu}=
\left\lbrace 
\begin{aligned}
1 \hspace{0.5cm} if \hspace{0.5cm}\mu&=\nu =0, 1,2,...,D_{+}\\
-1 \hspace{0.5cm} if \hspace{0.5cm} \mu&=\nu = D_{+}+1,  D_{+}+2,\cdots,D-1\\
0 \hspace{0.5cm} if \hspace{0.5cm} \mu &\neq \nu  \, \, .\\
\end{aligned}
\right.
\end{equation}
For the case of Minkowski's spacetime for instance, the signature is (+, -, -, -)=(1,3). The corresponding CCRs can be written in the following form : 
\begin{equation}
\left\lbrace 
\begin{aligned}
\left[\mathbf{ p}_{\mu}, \mathbf{x}_{\nu}\right]  &=\mathbf{ p}_{\mu} \mathbf{x}_{\nu} - \mathbf{x}_{\nu}\mathbf{p}_{\mu}
= \left\lbrace 
\begin{aligned}
i\hbar  \hspace{0.5cm} if \hspace{0.5cm}\mu&=\nu =0\\
-i\hbar  \hspace{0.5cm} if \hspace{0.5cm} \mu&=\nu =l=1,2,3\\
0 \hspace{0.5cm} if \hspace{0.5cm} \mu &\neq \nu\\
\end{aligned}
\right. \\
\left[\mathbf{ p}_{\mu}, \mathbf{p}_{\nu}\right]   &= \mathbf{ p}_{\mu} \mathbf{p}_{\nu} - \mathbf{p}_{\nu}\mathbf{p}_{\mu}=0\\
\left[\mathbf{ x}_{\mu}, \mathbf{x}_{\nu}\right]   &=\mathbf{ x}_{\mu} \mathbf{x}_{\nu} - \mathbf{x}_{\nu}\mathbf{x}_{\mu}=0 \, \, .
\end{aligned}
\right .
\end{equation}
Let us consider the example of a single pair $ (\textbf{p}, \textbf{x}) = \mathbf{({p_{1}}, {x_{1}})} $ with the metric $\eta=\eta_{11}=-1 $, the corresponding CCR is: 
\begin{equation}
\mathbf{\left[ p_{1}, x_{1}\right] }  =\mathbf{ p_{1} x_{1} - x_{1}p_{1}}=-i\hbar  \, \, .
\end{equation}
An advantage of this example is that we can highlight the distinction between covariant momentum and position operators $ \mathbf{p_{1}}$ and $\mathbf{ x_{1}} $(low index) and their contravariants analogous $ \mathbf{p^{1}}$ and $\mathbf{ x^{1}} $
(high index). One has the relation:
\begin{equation}
\left\lbrace 
\begin{aligned}
\mathbf{p_{1}}=&\eta_{11}\mathbf{p^{1}}=-\mathbf{p^{1}}\iff \mathbf{p^{1}}=-\mathbf{p_{1}}\\
\mathbf{x_{1}}=&\eta_{11}\mathbf{x^{1}}=-\mathbf{x^{1}}	\iff \mathbf{x^{1}}=-\mathbf{x_{1}}  \, \, .
\end{aligned}
\right.
\end{equation}
To reconcile the concept of phase space with the uncertainty principle, one should then consider quantum " momentum-coordinate joint states " that do not violate the uncertainty relation and are compatible with the CCRs. The best kind of states that satisfy these criteria  and saturate the uncertainty relations are the states denoted $ |\langle z_{1} \rangle \rangle  $  which correspond to Gaussian wave packets. These states $ |\langle z_{1} \rangle\rangle $ are also called phase space states. The expression of the coordinate wavefunctions which corresponds to this kind of states is \cite{highlighting}:
\begin{equation}
\langle x^{1}|\langle z_{1} \rangle \rangle =\left( \frac{1}{2\pi\mathcal{X}_{11}}\right) ^{1/4}e^{-\frac{\mathcal{B}_{11}}{\hbar^{2}}(x^{1}-\langle x^{1}\rangle)^{2}-\frac{i}{\hbar}\langle p_{1}\rangle x^{1}+ iK}  \, \, ,
\end{equation} 
where K is an arbitrary real number that is independent of the coordinates $ x^{1} $  ( $ e^{iK} $ is a unitary complex number) and 
\begin{equation}
\mathcal{B}_{11}=-i\frac{\varrho_{11}^{p}}{2\mathcal{X}_{11}}=\frac{\hbar^{2}}{4\mathcal{X}_{11}}-i\frac{\hbar\varrho_{11}}{2\mathcal{X}_{11}}= \frac{\mathcal{P}_{11}}{1+\frac{(\varrho_{11})^{2}}{4\hbar^{2}}}\left( 1-i\frac{\varrho_{11}}{2\hbar^{2}}\right) .
\end{equation}
The parameter $ \mathcal{B}_{11} $ in these relations is linked to the momenta-coordinates statistical variance-covariances, denoted $ \mathcal{P}_{11}, \mathcal{X}_{11}$ and $ \varrho_{11}$, corresponding to the state  $ |\langle z_{1} \rangle \rangle  $ itself. For any quantum state $\left| \Psi \right\rangle  $ , The momentum and coordinate mean values and statistical variance covariances are defined by the following relations  : 
\begin{equation}
\left\lbrace 
\begin{aligned}
\left\langle p_{1} \right\rangle &=\langle \psi |\mathbf{p}_{1}| \psi\rangle =- \langle \psi |\mathbf{p}^{1}| \psi\rangle = -\langle p^{1} \rangle\\
\left\langle x_{1} \right\rangle &=\left\langle \psi |\mathbf{x}_{1}| \psi\right\rangle = -\langle \psi |\mathbf{x}^{1}| \psi\rangle = -\langle p^{1} \rangle \\
\mathcal{P}_{11} &=\langle \psi |(\mathbf{p_{1}}-\langle p_{1}\rangle)^{2} | \psi\rangle = -\mathcal{P}_{1}^{1}=\mathcal{P}^{11} \\
\mathcal{X}_{11} &=\langle \psi |(\mathbf{x_{1}}-\langle x_{1}\rangle)^{2} | \psi\rangle = -\mathcal{X}_{1}^{1}=\mathcal{X}^{11} \\
\varrho_{11}^{p}&=\langle \psi |(\mathbf{p}_{1}-\langle p_{1}\rangle) \langle  (\mathbf{x}_{1}-\langle x_{1}\rangle)| \psi\rangle \\
\varrho_{11}^{x}&=\langle \psi |(\mathbf{x}_{1}-\langle x_{1}\rangle) \langle  (\mathbf{p}_{1}-\langle p_{1}\rangle)| \psi\rangle \\
\varrho_{11}&=\frac{1}{2}(\varrho_{11}^{p}+\varrho_{11}^{x})  \, \, .
\end{aligned}
\right.
\label{varhigh}
\end{equation}
Given the CCRs, the following relation is deduced 
\begin{equation}
\varrho_{11}=\frac{1}{2}(\varrho_{11}^{p}+\varrho_{11}^{x})= \varrho_{11}^{p}+i\frac{\hbar}{2}=\varrho_{11}^{x}-i\frac{\hbar}{2}  \, \, .
\end{equation}
The momentum-coordinate statistical variance-covariances matrix is given by
\begin{equation}
M_{11}=\begin{pmatrix}
\mathcal{P}_{11}& \varrho_{11}\\  
\varrho_{11} & \mathcal{X}_{11} 
\end{pmatrix} .
\end{equation}
By applying the scalar product properties in quantum state space along with the Cauchy-Schwarz inequality, it can be shown that the determinant of this matrix obeys the following relation \cite{invariant}:
\begin{equation}
\left| \begin{array}{cc}
\mathcal{P}_{11}& \varrho_{11}\\
\varrho_{11} & \mathcal{X}_{11}
\end{array}\right| 
= (\mathcal{P}_{11})(\mathcal{X}_{11})-(\varrho_{11})^{2}\geq \frac{\hbar^{2}}{4}  \, \, .
\end{equation}
This relation represents a rigorous form of the uncertainty principle that also  includes the statistical momentum-coordinate covariance $ \varrho_{11}  $.  The states $ |\langle z_{1} \rangle \rangle  $ saturate  the uncertainty relation, i.e. we have for $\left| \Psi \right\rangle  $ = $ |\langle z_{1} \rangle \rangle  $  :
\begin{equation}
\left| \begin{array}{cc}
\mathcal{P}_{11}& \varrho_{11}\\
\varrho_{11} & \mathcal{X}_{11}
\end{array}\right| 
= (\mathcal{P}_{11})(\mathcal{X}_{11})-(\varrho_{11})^{2}= \frac{\hbar^{2}}{4}\,\, . \label{eqmatrix}
\end{equation}
It can be shown that a state $ |\langle z_{1} \rangle \rangle =|\langle p_{1}\rangle, \langle x_{1}\rangle ,\mathcal{P}_{11} ,\varrho_{11}\rangle $ is an eigenstate of the operator 
\begin{equation}
\mathbf{ z_{1}}=\mathbf{ p_{1}}-\frac{2i}{\hbar}\mathcal{B}_{11}\mathbf{x^{1}} \, \, .
\end{equation} 
The corresponding eigenvalue equation is given by the following relation :
\begin{equation}
\mathbf{ z_{1}}|\langle p_{1}\rangle, \langle x_{1}\rangle ,\mathcal{P}_{11} ,\varrho_{11}\rangle=\left(  \langle p_{1}\rangle-\frac{2i}{\hbar}\mathcal{B}_{11}\langle x^{1}\rangle\right)|\langle p_{1}\rangle, \langle x_{1}\rangle ,\mathcal{P}_{11} ,\varrho_{11}\rangle=\langle z_{1}\rangle |\langle z_{1}\rangle \rangle \, \, .
\end{equation}
The eigenvalue of the operator $ \mathbf{ z_{1}} $ for a state $ |\langle p_{1}\rangle, \langle x_{1}\rangle ,\mathcal{P}_{11} ,\varrho_{11}\rangle $ is its mean value $  \langle z_{1}\rangle  =\langle p_{1}\rangle -\frac{2i}{\hbar}\mathcal{ B}_{11} \langle x^{1}\rangle $ itself. 
The joint momentum coordinate state $ |\langle z_{1} \rangle \rangle =|\langle p_{1}\rangle, \langle x_{1}\rangle ,\mathcal{P}_{11} ,\varrho_{11}\rangle $ is defined by the mean values $ \langle x_{1} \rangle $ and $ \langle p_{1} \rangle $ for a given  value of the statistical variance covariance matrix $M_{11} $.\\
\\
As in \cite{linear}, the quantum phase space can be defined as the set $ \left\lbrace (\left\langle p_{1}\right\rangle ,\left\langle x_{1}\right\rangle  )\right\rbrace   $ of the possible values of $ \langle p_{1}\rangle  $ and $ \langle x_{1}\rangle  $ for a given value of the statistical variance-covariance matrix $ M_{11} = \begin{pmatrix}
\mathcal{P}_{11}& \varrho_{11}\\  
\varrho_{11} & \mathcal{X}_{11} 
\end{pmatrix}  $  or equivalently as the set     $ \left\lbrace \left\langle z_{1}\right\rangle \right\rbrace  $ of all possible values of the mean values of the operator $ \mathbf{z_{1}} $ . The concept of quantum phase space can describe the natural relations between quantum mechanics, phase space, statistical mechanics and thermodynamics and may have interesting applications in this framework. One can for instance  apply this concept to establish new kinds of relativistic and quantum corrections for the Maxwell-Boltzmann ideal gas model as shown in the reference \cite{correction}.
The multidimensional and generalization of the concept of quantum phase space  are given in the reference \cite{highlighting, neutrino}.
\section{The Lindblad Model}
 The Lindblad theory provides a self-consistent
treatment of damping as a possible extension of quantum mechanics to open systems. The Lindblad master equation is the most general form of Markovian master equation for open quantum systems, describing irreversible and non-unitary dynamics. Whereas Markovian master equations valid under the Markov approximation( no memory of the environment's past), the Lindblad master equation is the specific, mathematically rigorous form within that class that ensures the dynamics remain completely positive and trace-preserving. \\
Using the structural theorem of Lindblad \cite{lindblad},
which gives the most general form of the bounded, completely dissipative Liouville operator $ L $, we obtain the explicit form of the most general time-homogeneous quantum mechanical Markovian master equation, as shown in reference \cite{isar}:
\begin{equation}
L(\bm{\rho}(t))=\frac{d \bm{\rho}(t)}{dt}=-\frac{i}{\hbar}[\bm{H}, \bm{\rho}(t) ]+ \frac{1}{2\hbar}\sum_{j}([\bm{V_{j}}\bm{\rho}(t),\mathbf{V^{\dagger}_{j}}]+[\bm{V_{j}},\bm{\rho}(t)\bm{V^{\dagger}_{j}}])\label{lindb1} \, \, .
\end{equation}
 
Here, $ \mathbf{H} $ is the Hamiltonian of the system in the absence of the environment. The operators $ \bm{V_{j}}, \bm{V_{j}^{\dagger}} $ are bounded operators on the Hilbert space $ \mathcal{H} $ of the Hamiltonian and they model the effect of the environment.\\
To get an exactly soluble model, a simple condition to the operators $ \mathbf{H}, \bm{V_{j}}, \bm{V_{j}^{\dagger}} $ should be imposed \cite{isar}: they are functions of the basic observables $ \mathbf{x} $ and $ \mathbf{ p} $ of the one-dimensional quantum mechanical system (with $ [\mathbf{x},\mathbf{p}]=i\hbar \mathbf{I} $, where $ \mathbf{I} $ is the identity operator on $ \mathcal{H} $) of the model. Then
the harmonic oscillator Hamiltonian $ \mathbf{H} $ is chosen of the general form\footnote{In the original references  (see, e.g., Ref.\cite{isar, scutaru}), the position operator is denoted by $ \hat{q} $; we use $ \mathbf{x} $ here to conform with the notation in Section 3.}  
\begin{equation}
\mathbf{H}=\frac{\mathbf{p}^{2}}{2m}+\frac{m\omega^{2}}{2}\mathbf{x}^{2}+\frac{\mu}{2}(\mathbf{px}+\mathbf{xp}) \, \, .
\end{equation} 
In this exact soluble model, as shown in \cite{ hasse}, one has to take $ \bm{V_{j}}  $ as a first-degree non-commutative polynomial in coordinate and momentum, and because $ \mathbf{p} $ and $ \mathbf{x} $ span the linear space of the first-degree polynomials in $ \mathbf{p} $ and $ \mathbf{x} $, there exist only two linearly independent operators $ \bm{V_{1}} $ and $ \bm{V_{2}} $, with $ \bm{V}_{i}=a_{i}\mathbf{p} + b_{i}\mathbf{x}, i=1,2 $ and $ a_{i}, b_{i}  $ are complex numbers.\\
With the following notations which introduces diffusion coefficients: 
\begin{equation}
D_{xx}=\frac{\hbar}{2}\sum_{j=1}^{2}|a_{j}|^{2},    D_{pp}=\frac{\hbar}{2}\sum_{j=1}^{2}|b_{j}|^{2}, D_{px}=D_{xp}=-\frac{\hbar}{2}\operatorname{Re}\sum_{j=1}^{2}a^{*}_{j}b_{j}
\end{equation}

$ L(\bm{\rho}) $ takes the following form 
\begin{equation}
\begin{split}
L(\bm{\rho})= &-\frac{i}{\hbar}[\mathbf{H},\bm{\rho}]-\frac{i(\lambda+\mu)}{2\hbar}[\mathbf{x},\bm{\rho}\mathbf{ p} +\mathbf{p} \bm{\rho} ]+\frac{i(\lambda-\mu)}{2\hbar}[\mathbf{p},\bm{\rho}\mathbf{x} +\mathbf{x} \bm{\rho} ] \\
&-\frac{D_{xx}}{\hbar^{2}}[\mathbf{p},[\mathbf{p},\bm{\rho}]]-
\frac{D_{pp}}{\hbar^{2}}[\mathbf{x},[\mathbf{x},\bm{\rho}]]+\frac{D_{px}}{\hbar^{2}}([\bm{p},[\mathbf{x},\bm{\rho}]]+[\mathbf{x},[\mathbf{p},\bm{\rho}]])\label{lindb2}
\end{split} 
\end{equation}
where \begin{equation}
\lambda= - \operatorname{Im}\left(  \sum_{j=1}^{2}a^{*}_{j}b_{j}\right)  \, \, ,
\end{equation}
represents the friction constant.\\ 
The diffusion coefficients satisfy the following fundamental constraints \cite{Isarsand, scutaru}:
\begin{equation}
D_{pp} >0 ,\hspace{1cm}
D_{xx}> 0 , \hspace{1cm}
 D_{xx}D_{pp}-(D_{xp})^{2}\geq \frac{\lambda^{2}\hbar}{4} \label{lindb3} \, \, .
\end{equation}
In the particular case when the asymptotic states is a Gibbs state 
\begin{equation}
\bm{\rho}_{G}(\infty)=\dfrac{e^{-\frac{\mathbf{H}_{0}}{kT}}}{Tr e^{-\frac{\mathbf{H}_{0}}{kT}}}
\end{equation}
these coefficients reduce to 
\begin{equation}
D_{pp}=\frac{\lambda + \mu}{2}\hbar m \omega \coth \frac{\hbar \omega}{2kT},\quad D_{xx}=\frac{\lambda - \mu}{2}\frac{\hbar}{m\omega} \coth \frac{\hbar \omega}{2kT},\quad  D_{px}=0 \, \, ,\label{Gibbs2}
\end{equation}
where T is the temperature of the thermal bath and $ \mathbf{ H_{0}}=\frac{\mathbf{p}^{2}}{2m}+\frac{m\omega^{2}}{2}\mathbf{x}^{2} $, the Hamiltonian of the free harmonic oscillator. Following the relation \eqref{Gibbs2}, the fundamental  constraints are satisfied only if $ \lambda > |\mu| $, i.e. 
\begin{equation}
D_{xx}D_{pp}-D_{px}^{2}=\frac{\lambda^{2}-\mu^{2}}{4} \hbar^{2}\coth^{2}\frac{\hbar \omega}{2kT}
\end{equation}\\
In the following we denote by $ \sigma_{AA} $ the dispersion of the operator $ \bm{A} $ and we donote by $ \sigma_{AB} $ the correlation of operators $ \mathbf{A} $ and $ \bm{B} $, i.e. :
\begin{align}
\sigma_{AA}=&\langle \bm{ A}^{2}\rangle - \langle \bm{A }\rangle ^{2} \label{sigmaA}\\
\sigma_{AB}=&\frac{1}{2}(\bm{AB}+\bm{BA})-\langle \bm{A} \rangle \langle \bm{B} \rangle \label{sigmaAA} \, \, .
\end{align}
 where $ \langle \bm{A} \rangle \equiv \sigma_{A}=Tr (\bm{\rho A})   \text{ and } Tr(\bm{\rho})=1$.
 The relations \eqref{sigmaA} and \eqref{sigmaAA}, it follows that:
 \begin{align}
 \frac{d\sigma_{A}(t)}{dt}=& TrL(\bm{\rho}(t))\bm{A} \label{sigmaA1}\\
 \frac{d\sigma_{AA}(t)}{dt}=& TrL(\bm{\rho}(t))\bm{A}^{2}-2\frac{d\sigma_{A}(t)}{dt}\sigma_{A}(t)\label{sigmaAA2} \, \, .
 \end{align}  
 Using the above definitions and the relation \eqref{varhigh}, we use the following notations :

\begin{equation}
\left\lbrace
\begin{aligned}
\sigma_{p}(t)&=Tr(\bm{\rho}(t)\mathbf{p}) \equiv  \left\langle p_{1} \right\rangle =\langle \psi |\mathbf{p}_{1}| \psi\rangle   \\
\sigma_{x}(t)&=Tr(\bm{\rho}(t)\mathbf{x})\equiv\left\langle x_{1} \right\rangle =\left\langle \psi |\mathbf{x}_{1}| \psi\right\rangle    \\
\sigma_{pp}(t)&=Tr(\bm{\rho}(t)\mathbf{p}^{2})-\sigma_{p}(t)^{2}\equiv\mathcal{P}_{11}(t) =\langle \psi |(\mathbf{p_{1}}-\langle p_{1}\rangle)^{2} | \psi\rangle \\
\sigma_{xx}(t)&=Tr(\bm{\rho}(t)\mathbf{x}^{2})-\sigma_{x}(t)^{2} \equiv\mathcal{X}_{11}(t) =\langle \psi |(\mathbf{x_{1}}-\langle x_{1}\rangle)^{2} | \psi\rangle  \\
\sigma_{px}(t)&=Tr(\bm{\rho}(t)(\frac{\mathbf{px}+\mathbf{xp}}{2})-\sigma_{x}(t)\sigma_{p}(t)\equiv \varrho_{11}(t)=\frac{1}{2}(\varrho_{11}^{p}(t)+\varrho_{11}^{x}(t)) \, \, .
\end{aligned}
\right.
\end{equation} 
The correlation matrix is given by the following relation :
\begin{equation}
\sigma(t)=\begin{pmatrix}
\sigma_{xx}(t)& \sigma_{px}(t)\\
\sigma_{px}(t) & \sigma_{pp}(t) 
\end{pmatrix} \, \, .
\end{equation} 
  with $  \det \sigma(t) \geq \frac{\hbar^{2}}{4}$ .\\ \\
The study of the time evolution of the momentum-coordinate statistical variance-covariances  matrix $ M_{11}(t) $ is equivalent to the study of the time evolution of the correlation matrix $ \sigma(t) $ as they have the same properties, quantum dynamics, and physical meaning, and  both satisfy the general uncertainty principle $ \det M_{11}(t)= \det\sigma(t) \geq \frac{\hbar^{2}}{4}$, with 
\begin{equation}
M_{11}(t)=\begin{pmatrix}
\mathcal{P}_{11}(t)& \varrho_{11}(t)\\
\varrho_{11}(t) & \mathcal{X}_{11}(t)
\end{pmatrix} \equiv \sigma(t)=\begin{pmatrix}
\sigma_{xx}(t)& \sigma_{px}(t)\\
\sigma_{px}(t) & \sigma_{pp}(t)
\end{pmatrix} \, \, .\label{correspondance}
\end{equation}
Equation \eqref{correspondance} expresses the correspondence between the quantum phase space and the Lindblad descriptions of the system. Physically, their equivalence means that the evolution of quantum fluctuations and correlations obtained from the Lindblad master equation can be directly interpreted in the language of the quantum phase space.

\subsection{Evolution of the covariance matrix}

According to \cite{scutaru}, it can be shown that there exist a matrix $ N $ with property $ N^{2}=I $ where $ I $ is the identity matrix   and a diagonal matrix $ K $ such that $ R=NKN $.
From this, the covariance matrix evolves in the course of the
time according to \cite{isar} :
\begin{equation}
X(t)=(Ne^{-Kt}N)X(0)-N(e^{-Kt}+I)K^{-1}ND \, \, ,
\end{equation}
with the vector X(t) and D are given by :
\begin{equation}
X(t)=
\left(  
\begin{aligned}
m\omega\sigma_{xx}(t)\\
\sigma_{pp}(t)/m\omega\\
\sigma_{px}(t)
\end{aligned}
\right) 
\text{ and }
D=
\left(  
\begin{aligned}
2m\omega D_{xx}\\
2D_{pp}/mw\\
2D_{px}
\end{aligned}
\right) .
\end{equation}
The existence of such state implies that 
\begin{equation}
\lim\limits_{t\rightarrow \infty} e^{-Kt}=0 \,\, .
\end{equation}
Therefore 
\begin{equation}
X(\infty)=(NK^{-1}N)D = R^{-1}D\,\, .
\end{equation} 
The integration of \eqref{sigmaA1} and \eqref{sigmaAA2} is straightforward (where $ A $ is set equal to $ p $ or $ x $, using the expression for $ L(\bm{\rho(t)}) $ from the master equation) . 
As shown in the reference \cite{scutaru}, there are two cases: the underdamped case ($ \omega > \mu $) and the overdamped case  ($ \mu > \omega $).\\
In the underdamped case ($ \omega > \mu $) the matrices $ N $ and $ K $ are given by:

	\begin{equation}
N=\frac{1}{2i\Omega} 
\begin{pmatrix}
\mu + i\Omega & \mu-i\Omega & 2\omega\\
\mu - i\Omega & \mu+i\Omega & 2\omega\\
-\omega & -\omega &-2\mu
\end{pmatrix} 
\end{equation}
and
\begin{equation}
K= 
\begin{pmatrix}
2(\lambda-i\Omega) & 0 & 0\\
0 & 2(\lambda+i\Omega) & 0\\
0 & 0 &2\lambda
\end{pmatrix} 
\end{equation}
with $ \Omega^{2}= \omega^{2}-\mu^{2} $ .\\ \\
In the overdamped case  ($ \mu > \omega $) the matrices $ N $ and $ K $ are given by 
	\begin{equation}
N=\frac{1}{2\nu} 
\begin{pmatrix}
\mu + \nu & \mu-\nu & 2\omega\\
\mu - \nu & \mu+\nu & 2\omega\\
-\omega & -\omega &-2\mu
\end{pmatrix} 
\end{equation}
and
\begin{equation}
K= 
\begin{pmatrix}
2(\lambda-\nu) & 0 & 0\\
0 & 2(\lambda+\nu) & 0\\
0 & 0 &2\lambda
\end{pmatrix} 
\end{equation}
with $ \nu^{2}= \mu^{2}-\omega^{2} $ .\\ \\
The relation $ X(\infty)=(NK^{-1}N)D = R^{-1}D $  is remarkable because it gives a very simple connection between the asymptotic values ( $ t\rightarrow \infty $) of $ \sigma_{xx}(t), \sigma_{pp}(t),\sigma_{px}(t) $ and the diffusion coefficients $ D_{xx},D_{pp},D_{px} $. As an immediate consequence of it, both underdamped and overdamped cases have the following explicit form \cite{isar}: 
\begin{align}
\sigma_{xx}(\infty)=&\frac{1}{2(m\omega)^{2}\lambda(\lambda^{2}+\omega^{2}-\mu^{2})}\left( (mw)^{2}(2\lambda(\lambda+\mu)+\omega^{2})D_{xx}+ \omega^{2}D_{pp}+2m \omega^{2}(\lambda+\mu)D_{px}\right) \\
\sigma_{pp}(\infty)=&\frac{1}{2\lambda(\lambda^{2}+\omega^{2}-\mu^{2})}\left( (m\omega)^{2}\omega^{2}D_{xx}+(2\lambda(\lambda-\mu)+\omega^{2})D_{pp}-2m\omega^{2}(\lambda-\mu)D_{px}\right) \\
\sigma_{px}(\infty)=&\frac{1}{2m\lambda(\lambda^{2}+\omega^{2}-\mu^{2})}\left( -(\lambda+\mu)(m\omega)^{2}D_{xx}+(\lambda-\mu)D_{pp}+2m(\lambda^{2}-\mu^{2})D_{px}\right) \,\, .
\end{align}
These relations show that the asymptotic values $ \sigma_{xx}(\infty), \sigma_{pp}(\infty), \sigma_{px}(\infty) $ do not depend on the initial values $ \sigma_{xx}(0), \sigma_{pp}(0), \sigma_{px}(0) $.
Conversely, if the relations\\ $ D=-RX(\infty) $ are considered, i.e. 
\begin{equation}
\begin{pmatrix}
2m\omega D_{xx}\\
\frac{2}{2m\omega}D_{pp}\\
2 D_{px}
\end{pmatrix}
 = \begin{pmatrix}
2(\lambda-\mu) & 0 & -2\omega\\ 
0 & 2(\lambda+\mu) & 2\omega\\
\omega & -\omega & 2\lambda
\end{pmatrix}  
\begin{pmatrix}
m\omega \sigma_{xx}(\infty) \\
\frac{1}{m\omega}\sigma_{pp}(\infty)\\
\sigma_{px}(\infty)
\end{pmatrix}\,\, .
\end{equation}
This yields: 
\begin{align}
D_{xx}&=(\lambda-\mu)\sigma_{xx}(\infty)-\frac{1}{m}\sigma_{px}(\infty) \notag\\
D_{pp}&=(\lambda+\mu)\sigma_{pp}(\infty)+m\omega^{2}\sigma_{px}(\infty) \label{then}\\ 
D_{px}&=\frac{1}{2}(m\omega^{2}\sigma_{xx}(\infty)-\frac{1}{m}\sigma_{pp}(\infty)+2\lambda\sigma_{px}(\infty))\,\, .\notag 
\end{align}
These expressions are crucial as they allow us to determine the specific diffusion coefficients $ D_{xx}, D_{pp}, D_{px} $ that are required to maintain a desired asymptotic state with variances $ \sigma_{xx}(\infty), \sigma_{pp}(\infty),\sigma_{px}(\infty) $. In the following section, we will use this to solve the inverse problem: finding the asymptotic state that corresponds to a pure time invariant pointer state.
\section{Identification of the joint momenta-coordinates as pointer states}\label{section 5}
\subsection{The generalized uncertainty}
Within the Lindblad framework for the damped harmonic oscillator, the relation \eqref{lindb3} represents a necessary condition for the  generalized uncertainty inequality to be fulfilled. Exploiting the fact that the linear positive mapping defined by $ \bm{A}\rightarrow Tr(\bm{\rho}\bm{A}) $ is completely positive, one can derive the following inequality \cite{scutaru}  :
\begin{equation}
D_{pp}\sigma_{xx}(t) + D_{xx}\sigma_{pp}-2D_{px}\sigma_{px}(t)\geq \frac{\hbar^{2}\lambda}{2}
\label{gen1}\,\, .
\end{equation}
This inequality, which must hold for all times $ t $, is equivalent to the generalized uncertainty inequality at any time $ t $, 
\begin{equation}
\sigma_{xx}(t)\sigma_{pp}(t)-\sigma_{px}^{2}(t)\geq\frac{\hbar^{2}}{4}
\label{gen2}\,\, ,
\end{equation}
provided that the initial values $ \sigma_{xx}(0) , \sigma_{pp}(0)$ and $ \sigma_{px}(0) $ for $ t=0 $ satisfy this inequality.\\
In particular, if the initial state corresponds to the ground state of the harmonic oscillator, then \cite{isar} 
\begin{equation}
\sigma_{xx}(0)= \frac{\hbar }{2m\omega}, \sigma_{pp}(0)=\frac{m\hbar\omega}{2}, \sigma_{px}(0)=0\,\, .
\end{equation}
\subsection{The predictability sieve criterion}
Zurek et \textit{al}. have earlier derived an approximate expression for the "predictability sieve" which is the measure of the increase in entropy\footnote{We use linear entropy in this paper.}, $ S_{l}= Tr(\bm{\rho}- \bm{\rho}^{2}) $, for a harmonic oscillator coupled to a heat bath whose dynamics is described by the Markovian Master equation.
Zurek's theory \cite{zurek} states that the pure states which minimize the entropy production in time, i.e. $ \dot{S}_{l}(t) =0 $ are the maximally predictive states.  These states remain least affected by the openness of the system and form the preferred set of states in the Hilbert space of the system, known as the $ \textit{pointer states} $. Their evolution is predictable with the principle of the least possible entropy production. By using the predictability sieve criterion, let us first look for the state which remain pure for all time.
\subsubsection{The state which remains pure for all time}
\begin{theorem}
	By using the complete positivity property of the dynamical semigroup $ \Phi_{t} $,  the relation :
\begin{equation}
Tr(\Phi_{t}(\bm{\rho})\sum_{j}\bm{V_{j}^{\dagger}}\bm{V_{j}}) =\sum_{j}Tr\left( \Phi_{t}(\bm{\rho})\bm{V_{j}^{\dagger}}\right) Tr\left( \Phi_{t}(\bm{\rho})\bm{V_{j}}\right)
\label{prop1} 
\end{equation}
represents the necessary and sufficient condition for $ \bm{\rho}(t)=\Phi_{t}(\bm{\rho}) $ to be a pure state for all $ t\geq $ 0. 	
\end{theorem}
	\begin{proof}
			A state is pure  if $ Tr(\bm{\rho}^{2})=Tr(\bm{\rho}) = 1 $, for all time t $ \geq 0 $. This is equivalent to
		\begin{equation}
		\frac{d}{dt}Tr(\bm{\rho}(t)^{2}) = 2 Tr[\bm{\rho}(t) L(\bm{\rho}(t))]=0,  \text{for all $ t\geq $ 0} \label{pure1}\,\, .
		\end{equation}
		Using the explicit form of $ L(\bm{\rho}(t)) $ it follows that 
		\begin{equation}
		\frac{d}{dt}Tr(\bm{\rho}(t)^{2}) =\frac{1}{2\hbar}\sum_{j} \left( Tr (\bm{\rho}(t)\bm{V_{j}}\bm{\rho}(t)\bm{V_{j}^{*}})-Tr(\bm{\rho}(t)^{2}\bm{V_{j}^{*}}\bm{V_{j}})\right) =0\,\, .
		\end{equation}
		For a pure state $ \bm{\rho}^{2}(t)=\bm{\rho}(t) $ and $ \bm{\rho}(t)\bm{A\rho}(t)=Tr (\bm{\rho}(t)\bm{A})\bm{\rho}(t) $ for any $\bm{ A} \in \mathcal{ B}(\mathcal{ H}) $. Then \eqref{pure1} can be written as 
		\begin{equation}
		Tr( \Phi_{t}(\bm{\rho})\sum_{j}\bm{V_{j}^{\dagger}V_{j}})  =\sum_{j}Tr\left( \Phi_{t}(\bm{\rho})\bm{V_{j}^{\dagger}}\right) Tr\left( \Phi_{t}(\bm{\rho})\bm{V_{j}}\right) \,\, .
		\end{equation}
	\end{proof}
This equality is a generalisation of the purity condition to all Markovian master equations \cite{scutaru, isar}. Importantly, it also provides the dynamical justification for the equality sign in the generalized uncertainty relation \eqref{eqmatrix}: a state that remains pure for all times must necessarily saturate the uncertainty bound. In this sense, the complete-positivity condition \eqref{prop1} ensures that the determinant of the momentum-coordinate covariance matrix attains its minimum value, thereby identifying the joint momenta-coordinates states as the only minimum uncertainty pure states compatible with Lindblad dynamics.
\begin{theorem}
For environment operators $ \bm{V_{j}} $ of the form \eqrefc{lindb1}, the pure state condition \eqref{prop1} takes the following form, which corresponds to equality in the relation \eqref{gen1}
\begin{equation}
D_{pp}\sigma_{xx}(t) + D_{xx}\sigma_{pp}(t)-2D_{px}\sigma_{px}(t)= \frac{\hbar^{2}\lambda}{2}
\end{equation}
and the generalized uncertainty relation \eqref{gen2} becomes into the following minimum uncertainty equality for pure states 
\begin{equation}
(\mathcal{P}_{11}(t))(\mathcal{X}_{11}(t))-(\varrho_{11}(t))^{2}=\sigma_{xx}(t)\sigma_{pp}(t)-\sigma_{px}^{2}(t)=\frac{\hbar^{2}}{4}\,\, .
\end{equation}	
\end{theorem}
\begin{proof}
	By eliminating $ \sigma_{pp} $ between the two equalities \eqrefc{gen1} and \eqref{gen2}, as in \cite{dekker}, the following relations which have to be fulfilled at any moment of time are obtained :
\begin{align}
&D_{pp}D_{xx}-D_{px}^{2}=\frac{\hbar^{2}\lambda^{2}}{4} \label{prop5}\\
&D_{pp}\sigma_{xx}(t)-D_{px}\sigma_{px}(t)-\frac{\hbar^{2}\lambda}{4}=0\\
&\sigma_{px}(t)\left( D_{pp}D_{xx}-D_{px}^{2}\right) -\frac{\hbar^{2}\lambda}{4}D_{px}=0 \label{prop7}\,\, .
\end{align}
Inserting the relation \eqrefc{prop5} into \eqrefc{prop7}: 
\begin{equation}
\sigma_{px}(t)\left( \frac{\hbar^{2}\lambda^{2}}{4}\right) -\frac{\hbar^{2}\lambda}{4}D_{px}=0 \Rightarrow \sigma_{xx}(t)= \frac{D_{xx}}{\lambda}\label{prop8}
\end{equation} 
Using \eqrefc{prop7} and \eqrefc{prop8} , then the pure states remain pure for all times only  if their variances-covariance  are constant in time and have the form: 
\begin{equation}
\sigma_{xx}(t)=\frac{D_{xx}}{\lambda} ,\quad \sigma_{pp}(t)=\frac{D_{pp}}{\lambda},\quad \sigma_{px}(t)=\frac{D_{px}}{\lambda}\label{prop9}\,\, .
\end{equation} 
Substituting the time-independent variances from  \eqrefc{prop9} into the expressions for the diffusion coefficients in  \eqrefc{then} yields the following constraints :
\begin{equation}
m\omega D_{xx}=\frac{1}{m\omega}D_{pp}, D_{px}=-m\mu D_{xx}=-\frac{\mu}{m\omega^{2}}D_{pp} \label{prop10}\,\, .
\end{equation}
Using the relation \eqref{prop10}, the relation \eqref{prop5} implies that :
\begin{equation}
m^{2}(\omega^{2}-\mu^{2})D_{xx}^{2}=\frac{\lambda^{2}\hbar^{2}}{4}\Rightarrow  D_{xx}=\frac{\lambda\hbar}{2m\Omega}\label{prop11}\,\, ,
\end{equation}
which can be satisfied only if $ \omega > \mu $ (underdamped case). We have denoted \\ 
$ \Omega^{2}=\omega^{2}-\mu^{2} $  . From equation \eqref{prop11} into \eqref{prop10}, we have these following relations:
\begin{equation}
D_{xx}=\frac{\lambda\hbar}{2m\Omega},\quad D_{pp}=\frac{\hbar\lambda m \omega^{2}}{2\Omega},\quad D_{px}=-\frac{\hbar \lambda \mu}{2\Omega} \label{prop12}\,\, .
\end{equation}
From \eqref{prop9} and \eqref{prop12}, we have obtained the following expression of variances which ensure that the initial pure states remain pure for any time $ t $:
\begin{equation}
\sigma_{xx}=\frac{\hbar}{2m\Omega},\quad \sigma_{pp}=\frac{\hbar m \omega^{2}}{2\Omega},\quad \sigma_{px}=-\frac{\hbar  \mu}{2\Omega}\,\, .
\end{equation}
Then the generalized uncertainty relation becomes :
\begin{equation}
(\mathcal{P}_{11})(\mathcal{X}_{11})-(\varrho_{11})^{2}=\sigma_{xx}\sigma_{pp}-\sigma_{px}^{2}=\frac{\hbar^{2}\omega^{2}}{4\Omega^{2}}-\frac{\hbar^{2}\mu^{2}}{4\Omega^{2}}=\frac{\hbar^{2}}{4}\frac{\omega^{2}-\mu^{2}}{\Omega^{2}}=\frac{\hbar^{2}}{4}\label{last1}\,\, .
\end{equation}
\end{proof}
This confirms that the states remain pure for all time if they satisfy this minimum uncertainty condition. In the overdamped regime ($ \mu > \omega $), the dissipation overwhelms the restorative potential. The mathematical consequence is that eq \eqref{prop11} admits no physical solution, meaning no pure state can possess time-independent variances that satisfy the master equation. Physically, this indicates that the environmental coupling is too violent to support any stable, minimal uncertainty wave-packet.  Now, let us look for the states which produce the least entropy.
\subsubsection{The states which produce the least entropy}
\begin{theorem}
	The states which produce the least entropy saturate the generalized uncertainty relations, that is  
	\begin{equation}
	(\mathcal{P}_{11})(\mathcal{X}_{11})-(\varrho_{11})^{2}=\sigma_{xx}(t)\sigma_{pp}(t)-\sigma_{px}^{2}(t)=\frac{\hbar^{2}}{4}\,\, .
	\end{equation} 
	 	\end{theorem}
	\begin{proof}
		The linear entropy is given by 
		\begin{equation}
		S_{l}(t)=Tr(\bm{\rho}-\bm{\rho}^{2})=1-Tr(\bm{\rho}^{2})\,\, .
		\end{equation}
The time derivative of the above equation gives
\begin{equation}
\dot{S}_{l}(t)=-2Tr(\bm{\rho} \dot{\bm{\rho}})=-2Tr(\bm{\rho} L(\bm{\rho}))\,\, ,
\end{equation}
where L is the evolution operator. Using the expression of L in the equation \eqref{lindb2}, we have :
\begin{equation}
\begin{split}
\dot{S}_{l}(t)=&\frac{4}{\hbar^{2}}[ D_{pp}Tr(\bm{\rho}^{2}\mathbf{x}^{2}-\bm{\rho} \mathbf{x}\bm{ \rho}\mathbf{ x}) +D_{xx} Tr(\bm{\rho}^{2}\mathbf{p}^{2}- \bm{\rho}\mathbf{ p} \bm{\rho}\mathbf{ p}) \\ 
& -D_{px}Tr (\bm{\rho}^{2}(\mathbf{xp} +\mathbf{px}) -2\bm{\rho}\mathbf{ x}\bm{ \rho}\mathbf{ p})-\frac{\hbar^{2}\lambda}{2}Tr(\bm{\rho}^{2})  ] \,\, .
\end{split}
\end{equation}
For a  state that remains approximately pure $( \bm{\rho}^{2} \approx \bm{\rho})  $, we obtain
\begin{equation}
\dot{S}_{l}(t)=\frac{4}{\hbar^{2}}(D_{pp}\sigma_{xx}(t)+D_{xx}\sigma_{pp}(t)-2D_{px}\sigma_{px}(t)-\frac{\hbar^{2}\lambda}{2})\geq 0\,\, .
\end{equation}  
We see that $ \dot{S}_{l}(t) =0 $ if 
\begin{equation}
D_{pp}\sigma_{xx}(t)+D_{xx}\sigma_{pp}(t)-2D_{px}\sigma_{px}(t)=\frac{\hbar^{2}\lambda}{2}\,\, ,
\end{equation}
which implies that 
\begin{equation}
(\mathcal{P}_{11})(\mathcal{X}_{11})-(\varrho_{11})^{2}=\sigma_{xx}(t)\sigma_{pp}(t)-\sigma_{px}^{2}(t)=\frac{\hbar^{2}}{4} \label{last}\,\, .
\end{equation}
thus satisfying the condition of purity at all times.
\end{proof}
 The satisfaction of this equality is the definitive signature of a pointer state. Therefore, the joint momenta-coordinates states are unequivocally identified as the preferred, robust states that emerge from the quantum-to-classical transition in the underdamped harmonic oscillator. 
\section{Discussions and conclusions}
The question of how the classical world emerges from quantum mechanics has long been a subject of debate. In a sense, decoherence explains the washing away of quantum coherences and the emergence of a state which makes classical sense. We can say that decoherence explanation  takes away some of the mystery from the idea of \textit{wave function collapse} and
provides a conventional mechanism to explain the appearance of a classical world.

In \cite{highlighting}, it was conjectured that pointer states could be identified with joint momenta- coordinates states, but no formal proof was offered. In this work, we have resolved that conjecture. We advanced the theory in two key ways. First, we extended the analysis to both underdamped and overdamped regimes, showing that while joint momenta-coordinates states remain pure and robust in the underdamped case, no such pure pointer states survive in the overdamped regime \eqref{prop11}. This clarifies the physical distinction between the two dynamical regimes and provides a more complete picture of decoherence in open systems. The overdamped regime's strong dissipation prevents the dynamical stability required for pointer state formulation, explaining why studies like Isar's naturally focused on oscillatory dynamics. Second, we reinterpreted the problem within the quantum phase space formalism, proving rigorously that joint momenta-coordinates states which saturate generalized uncertainty  relations are the genuine pointer states selected by the predictability-sieve criterion \eqref{last}. 

This resolution of the conjecture reinforces the conceptual bridge between decoherence theory and the concept of quantum phase-space. By rigorously identifying pointer states with minimum uncertainty wave packets in phase space, we demonstrate that classicality emerges precisely from those states that are maximally localized and informationally complete, maximizing predictability while minimizing dispersion and decoherence. Our results extended the concept of phase space decoherence introduce by Brody et \textit{al}. \cite{brody} who demonstrate that phase space decoherence drives any state toward a mixture of coherent states. Our work explains why these particular states emerge as the robust pointer states under joint measurement dynamics. Our findings may contribute to broader foundational questions, including Penrose's proposal of \textit{gravitizing quantum mechanics}\cite{gravitizing}.  It may also has practical implications: joint momenta-coordinates states emerge as natural candidates for stable qubits, resilient against decoherence. Such states may be directly relevant to the design of error-resistant quantum information protocols, which constitutes the main direction of our future work.

\end{document}